\documentclass[11pt]{article}
\usepackage{amsmath,amssymb}
\usepackage{amsthm}
\usepackage{graphicx}
\usepackage{float}
\usepackage{cases}
\newtheorem{theoremx}{Theorem}%[section]
\newtheorem{corollaryx}{Corollary}%[section]
\newtheorem{remark}{Remark}%[section]
\newtheorem{lemma}{Lemma}%[section]
\begin{document}
\title{\textbf{Time-varying Formation Tracking of Multiple Manipulators via Distributed Finite-time Control}
\thanks{Citation Information: Ming-Feng Ge, Zhi-Hong Guan, Chao Yang, Tao Li, \& Yan-Wu Wang. Time-varying formation tracking of multiple manipulators via distributed finite-time control. Neurocomputing, 2016, 202: 20-26. http://dx.doi.org/10.1016/j.neucom.2016.03.008}
 }
\vspace{1cm}
\author{\textrm{\textbf{Ming-Feng Ge\thanks{First author. Email: fmgabc@163.com (M.-F Ge).}, Zhi-Hong Guan\thanks{Corresponding author. Email: zhguan@mail.hust.edu.cn (Z.-H Guan).}, Chao Yang, Tao Li, Yan-Wu Wang}} \\ \\
%EndAName
{\small \textsl{College of Automation}}\\
{\small \textsl{Huazhong University of Science and Technology,
Wuhan,
430074, P. R. China}}\\
 }
\date{}
\maketitle

\begin{abstract}
Comparing with traditional fixed formation for a group of dynamical systems,
time-varying formation can produce the following benefits:
i) covering the greater part of complex environments; ii) collision avoidance.
This paper studies the time-varying formation tracking for multiple manipulator systems
(MMSs) under fixed and switching directed graphs with a dynamic leader, whose acceleration cannot change too fast.
An explicit mathematical formulation of time-varying formation is developed
based on the related practical applications.
A class of extended inverse dynamics control algorithms combining with distributed sliding-mode estimators are developed to
address the aforementioned problem.
By invoking finite-time stability arguments, several novel criteria (including sufficient criteria, necessary and sufficient criteria)
for global finite-time stability of MMSs are established.
Finally, numerical experiments are presented to verify the effectiveness of the theoretical results. \vspace{1em}

 \vspace{1em}
\end{abstract}

\section{ Introduction}

Recently, distributed control problems for a group of dynamical systems have attracted much attentions
due to its wide applications, including
coordination for multi-agent systems {\cite{WSSSW}}-{\cite{LiuZW01}},
synchronization in complex networks {\cite{LuYu,LuYu1}},
distributed computing in sensor networks \cite{Serpen}-\cite{OlfatiS},
multi-fingered hand grasping and manipulation \cite{Ueki,fmg03}.
Formation control is a significant issue in the distributed control field.
A formation is defined as a special configuration ($i.e.$, desired positions and orientations)
formed by a cluster of interconnected autonomous agents, in which a global goal is achieved cooperatively \cite{BSCW}.
Many formation control methods have been developed, such as
virtual structure methods {\cite{fmg19}}, behavior-based methods {\cite{fmg22,fmg25}},
leader-follower methods {\cite{Mahmood}}, artificial potential field methods {\cite{GW}}.
The aforementioned methods can only produce fixed formations for multi-agent systems.
However, in a number of real-world applications,
the formation of multi-agent systems is always changing to adapt to the dynamical changing environment.
It follows that the fixed formations cannot satisfy
the practical requirements of many real-world applications.
It thus motivates several research on time-varying formations.
Time-varying formation control algorithms for a group of unmanned aerial vehicles
with its applications to quadrotor swarm systems had been presented based on consensus theory {\cite{DYSZ}}.
Coherent formation control of a set of agents, including unmanned aerial vehicles and unmanned ground vehicles,
in the presence of time-varying formation had been studied in {\cite{RAN}}.
Time-varying formation implies that the formation of a multi-agent system can be changing as required
without losing system stability, which products the following benefits:
i) covering the greater part of complex environments; ii) collision avoidance.
However, to the authors' knowledge, the mathematical formulations of time-varying formation tracking are still not clear,
which impedes the development and applications of the relative technologies.

On the other hand, networked robotic systems have been broadly studied due to their various advantages,
including flexibility, adaptivity, fault tolerance, redundancy, and the possibility to
invoke distributed sensing and actuation \cite{AACM}.
Many control algorithms for global asymptotic tracking of networked robotic systems described by
Euler-Lagrange systems can be found in the literature.
Adaptive control approaches are proposed to address the leader-follower and
leaderless coordination problems for multi-manipulator systems based on graph theory {\cite{CLHT,CLZGHT}}.
Distributed containment control had been developed for global asymptotic stability of
Lagrangian networks under directed topologies containing a spanning tree {\cite{Mei01}}.
Some distributed average tracking algorithms had been developed invoking extended {PI} control
and applied to networked Euler-Lagrange systems {\cite{ChenFeng}}.
The task-space tracking control problems of networked robotic systems under strongly connected graphs
without task-space velocity measurements had been investigated \cite{LiangWang}.
In presence of kinematic and dynamic uncertainties, task-space synchronization had been addressed for multiple manipulators
under strong connected graphs by invoking passivity control {\cite{Wang01}} and adaptive control {\cite{Wang02}}.
All of the aforementioned control algorithms produce global asymptotic tracking
of robotic manipulators, which implies that the system trajectories
converge to the equilibrium as time goes to infinity.
Finite-time stabilization of dynamical systems may give rise to
fast transient and high-precision performances besides finite-time
convergence to the equilibrium, and a lot of work has been done in
the last several years \cite{XZFXT}-\cite{HWWS}.

Motivated by our preliminary work on distributed control {\cite{Guan01,Guan02}},
the time-varying formation tracking of multiple manipulator systems (MMSs) is taken into account.
Distributed finite control is developed to drive the centroid of the MMS to follow the leader
at a distance and to achieve the desired time-varying formation of the MMS meanwhile.
The main contributions are summarized as following:
i) Comparing with the existing work based on multi-agent systems with
single-integrator and double-integrator dynamics \cite{AACM},
we consider {MMSs} described by {Euler-Lagrange} systems.
ii) Comparing with the existing fixed formation tracking algorithms for multi-agent system {\cite{Yoo}},
we consider the time-varying formation tracking problems with a dynamic leader and present an explicit mathematical
formulation of time-varying formation based on its practical characteristics.
iii) Some novel estimator-based finite-time control algorithms are developed for the above time-varying formation tracking problems.
For the presented control algorithms, some conditions (including {sufficient conditions}, {necessary and sufficient conditions})
are derived to guarantee the achievement of time-varying formation tracking.

The rest of this paper is organized as follows: system formulation and some preliminaries are presented in Section 2.
The control algorithms and conditions of time-varying formation tracking are given in Section 3.
In Section 4, the simulation results are presented. The conclusions are provided in Section 5.

\section{ Preliminaries}
\subsection { System formulation}

The dynamics of the \textit{i}th manipulator in the {MMS} is given as following \cite{Lewis}:
\begin{equation}\label{1.1}
{\mathcal H}_i(q_i){\ddot q}_i + {\mathcal C}_i(q_i,\dot q_i){\dot q_i} + g_i(q_i) = \tau_i,
\end{equation}
where $i \in {\mathcal V} = \{ 1,2,\cdots,n \}$, $t \in {\mathcal J} = [t_0,\infty)$, $t_0 \geq 0$ is the initial time,
$q_i, \dot q_i$ and $\ddot q_i \in {\mathbb R}^m$ are the position, velocity and acceleration vectors of generalized coordinates,
${\mathcal H}_i(q_i)$ and ${\mathcal C}_i(q_i,\dot q_i) \in {\mathbb R}^{m \times m}$ are the inertia and the Coriolis/centrifugal force matrices,
$g_i(q_i)$ and $\tau_i \in {\mathbb R}^m$ denote the gravitational torque and the input torque respectively.

The leader for the MMS is given as following:
\begin{equation*}
\left\{\begin{array}{lll}
  \dot x_0 = v_0,\\
  \dot v_0 = a_0,
  \end{array} \right.
\end{equation*}
where $x_0,v_0,a_0 \in {\mathbb R}^m$ are the position, velocity and acceleration vectors of generalized coordinates respectively.

We invoke a directed graph ${\mathcal G} = \{ \mathcal V, \mathcal E, \mathcal W \}$ to describe the interaction of the MMS,
where $\mathcal V$ denotes the node set given right after (\ref{1.1}), $\mathcal E \subseteq {\mathcal V}^2$ is the edge set,
$\mathcal W = [w_{ij}]_{n \times n}$ represents the adjacency matrix.
The \textit{i}th node denotes the \textit{i}th manipulator in the MMS.
An edge $\{ j,i \} \in {\mathcal E}$ denotes that the \textit{i}th node can access information from the \textit{j}th node.
The adjacency weight $w_{ij}$ is defined as $w_{ij} > 0$ if $\{ j,i \} \in {\mathcal E}$, and $w_{ij} = 0$ otherwise.
Besides, self-edges are not allowed in this paper, $i.e.$, $w_{ii} = 0$.
A directed path from the \textit{i}th node to the \textit{j}th node is an ordered sequence of edges $\{ i_1,i_2 \}, \{ i_2,i_3 \}, \cdots,$
in the directed graph.
The neighbor set of the \textit{j}th manipulator is denoted by ${{\mathcal N}_i} = \{ {j \in {\mathcal V}~|~(j,i) \in \mathcal E} \}$.
${\mathcal G}$ is said to be undirected if and only if $\{ j,i \} \in {\mathcal E} \Leftrightarrow \{ i,j \} \in {\mathcal E}$,
$i.e.$, $w_{ij} = w_{ji} $, $\forall i,j \in \mathcal V$.
Throughout this paper, ${\mathcal G}$ is supposed to be undirected.
Let $\mathcal{P} = [{p_1},{p_2},\ldots,{p_n}]^T$ be the nonnegative weight vector between the \textit{n} nodes and the leader,
where ${p_i} > 0$ if the information of the leader is available to the \textit{i}th node,
namely, the \textit{i}th node is pinned; ${p_i} = 0$ otherwise.
The {Laplacian} matrix ${\mathcal L}$ for ${\mathcal G}$ is defined as
$l_{ii} = \sum \nolimits_{j=1}^n w_{ij}$ and $l_{ij} = - w_{ij}, i\neq j$.
Two assumptions throughout this paper are presented as following: \vspace{0.2cm}\\
\emph{\textbf{A1)} The leader has a directed path to the nodes in the set $\mathcal V$ under $\mathcal{G}$ and $\mathcal{P}$}; \\
\emph{\textbf{A2)} $\mathop {\sup }\nolimits_{t \in \mathcal J} \| \dot a_0(t) \| < \beta$,
where $\| \cdot \|$ represents the Euclidean norm and $\beta$ is a positive constant.} \vspace{0.2cm}\\
By {Assumption \emph{\textbf{A2}}}, the derivative of the acceleration $a_0(t)$ of the leader is bounded,
which happens to be the actual characteristics of the trajectories that can be reachable by the manipulators
described by Euler-Lagrange system \cite{MengDimarogonas}.

\begin{lemma}{\cite{fmg09}}\label{l5}
Suppose that Assumption \textbf{A1} holds. $\mathcal{M} = (\mathcal{L} + diag \{ \mathcal{P} \}) \otimes I_m \in {\mathbb R}^{mn \times mn}$ is symmetric positive definite,
where $\otimes$ denotes the Kronecker product and $I_m \in {\mathbb R}^{m \times m}$ represents the identity matrix.
\end{lemma}

\subsection {Problem Statement}

In a number of real-world applications, the desired formation for the MMS is required to be time-varying and switching according to task demands.
In this section, the explicit mathematical definition of time-varying formation tracking is presented.

Let ${\digamma_{0-k}} = \{ {{\digamma _0},{\digamma _1}, \ldots ,{\digamma _k}} \}$ be a finite set of desired formations,
where $\digamma _s = \{ \eta_{s1},\eta_{s2},\ldots,\eta_{sn} \}$ denotes the \textit{s}th desired formation,
$\eta_{si} \in {\mathbb R}^m$ denotes the local coordinate of the \textit{i}th manipulator in the
\textit{m}-dimensional Euclidean space with respect to $\digamma _s$, $\forall s = 0, 1, \ldots, k$.
Note that $\digamma _s$ becomes a desired geometric pattern in 2D plane if $m = 2$.
Let ${\mathcal I} = \{ 0,1,\ldots,k \}$ denote the index set of ${\digamma_{0-k}}$.
A switching signal ${\sigma(t)}: \mathcal J \to \mathcal I$ is introduced with a sequence of time points $\{ t_1,t_2,\ldots,t_s,\ldots \}$,
satisfying ${t_0} < {t_1} <  \cdots   < {t_s} <  \cdots$, at which the desired formation changes.
Let $\digamma(t)$ be the desired formation at time $t$.
Then for any $t \in [t_s,t_{s+1})$, the desired formation $\digamma(t) = \digamma_{\sigma(t)} = \digamma _s \in \digamma _{0-k}$.
Besides, we assume that the desired formation is closed at each time instant, $i.e.$,
$\sum\nolimits_{i = 1}^n {{\eta _{si}} = 0}$, $\forall s \in \mathcal I$.

The \textbf{control objective} is to design distributed control $\tau_i$ for the \textit{i}th manipulator by invoking
its information ($i.e.$, $q_i$, $\dot q_i$ and $\eta_{si}$) and its neighbour node's states
($i.e.$, $q_j$, $\dot q_j$ and $\eta_{sj}$ for $j \in {\mathcal N}_i$) such that for any $t \in [t_s,t_{s+1})$,
the \textbf{time-varying formation tracking} is said to be achieved for the {MMS}, $i.e.$,
\begin{equation}\label{1.2}
\left\{\begin{array}{lll}
\mathop {\lim }\limits_{t \to {t_f^s}}  \left \| {q_i} - {q_j} - \eta_{si} + \eta_{sj} \right \|  =  0, \\ \\
\mathop {\lim }\limits_{t \to {t_f^s}}  \left \| \frac{1}{n} \sum\limits_{i = 1}^n {q_i}  - {x_0} \right \| = 0,\\ \\
\mathop {\lim }\limits_{t \to {t_f^s}}  \left \| \dot q_i  - {v_0} \right \| = 0,
\end{array}\right.
\end{equation}
where $t_f^s$ denotes the settle time.
In this paper, we assume that the minimum switching interval $h = \mathop {\min }\limits_s  (t_{s+1} - t_s)$ is large
enough such that $t_f^s$ can be included in the half-open interval $[t_s, t_{s+1})$, $\forall s = 0,1,\cdots$.

\begin{remark}
Note that (\ref{1.2}) means that for any $[{t_s},{t_{s + 1}})$, the MMS converge to the desired formation $\digamma _s$
and the centroid of the MMS follows the leader before time $t_{s + 1}$.
By designing time-varying formations, the obstacle and collision avoidance can be achieved
while the centroid follows the leader.
It is worthy to point out that the control problem addressed in \cite{Yoo}
is a special case of (\ref{1.2}).
\end{remark}

\subsection {Finite-time stability}
Some concepts for {finite-time stability} and {homogeneous} systems are introduced in this section \cite{fmg26}.
Consider a \textit{k}-dimensional system
\begin{equation}\label{1.3}
\dot z = f(z), ~f(0) = 0, ~z(t_0) = z_0, ~z \in {{\mathbb R} ^k},
\end{equation}
where $k$ is an arbitrary positive integer.
The continuous vector field $f(z) = {\rm col}(f_1(z),f_2(z),\ldots,f_k(z))$ is {homogeneous} of degree
$\lambda \in \mathbb R$ with dilation $(\gamma _1,\gamma _2, \ldots ,\gamma _k)$,
if for any $\varepsilon  > 0$,
\begin{equation*}
  {f_i}({\varepsilon ^{{\gamma _1}}}{z_1},{\varepsilon ^{{\gamma _2}}}{z_2}, \ldots ,\varepsilon ^{\gamma _k} {z_k})
  = {\varepsilon ^{\lambda  + {\gamma _i}}}{f_i}(z),
\end{equation*}
where $i = 1,2 \ldots ,k$.
System (\ref{1.3}) is said to be homogeneous if its vector field is homogeneous.
Additionally, the following \textit{k}-dimensional system
\begin{equation}\label{1.4}
\dot z = f(z)+ \tilde f(z), \ \tilde f(0) = 0,
\end{equation}
is called being {locally homogeneous} of degree $\lambda \in \mathbb R$ with dilation $(\gamma _1,\gamma _2, \ldots ,\gamma _k)$,
if system (\ref{1.3}) is homogeneous and the continuous vector field $\tilde f(z)$ satisfies
\begin{eqnarray*}
\mathop {\lim }\limits_{\varepsilon  \to 0} \frac{{{\tilde f}_i}({\varepsilon ^{{\gamma _1}}}{z_1},{\varepsilon ^{{\gamma _2}}}{z_2},
\ldots ,{\varepsilon ^{{\gamma _k}}}{z_k})}{{{\varepsilon ^{\lambda  + {\gamma _i}}}}} = 0, ~\forall z \neq 0, i = 1,2 \ldots ,k.
\end{eqnarray*}
Based on the above presentations, some results and lemmas in \cite{fmg26}-\cite{Rosier} which will be used in this paper are proposed here.
\begin{lemma}\label{l1}
{(LaSalle's Invariance Principle)} Let $z(t)$ be a solution of
$\dot z = f(z)$, $z(t_0) = {z_0} \in {{\mathbb R} ^k},$
where $t_0$ is the initial time, $f:U \to {{\mathbb R} ^k}$ is continuous with an open subset $U$ of ${\mathbb R} ^k$,
and $V: U \to {\mathbb R}$ be a locally Lipschitz function such that ${D^ + }V(z(t)) \leq 0$, where ${D^ + }$ denotes the upper {Dini} derivative.
Then ${\Theta ^ + }({z_0}) \cap U$ is contained in the union of all solutions that remain in $S = \{ z \in U:{D^ + }V(z) = 0\} $,
where ${\Theta ^ + }({z_0})$ denotes the positive limit set.
\end{lemma}

\begin{lemma}\label{l2}
Suppose that system (\ref{1.3}) is {homogeneous} of degree $\lambda \in \mathbb R$ with dilation
$({\gamma _1},{\gamma _2}, \ldots ,{\gamma _k})$, $z = 0$ is its asymptotically stable equilibrium.
If homogeneity degree $\lambda < 0$, the equilibrium of system (\ref{1.3}) is finite-time stable.
Moreover, if system (\ref{1.4}) is locally homogeneous,
the equilibrium of system (\ref{1.4}) is locally finite-time stable.
\end{lemma}

\begin{lemma}\label{l3}
If the equilibrium of a closed-loop system is global asymptotic stable and local finite-time stable,
then it is also global finite-time stable.
\end{lemma}

\section{Time-varying formation tracking of multiple manipulators}

In this section, we are concerned with the time-varying formation tracking problems where the formations of the {MMS}
is time-varying and the leader has varying vectors of generalized coordinate derivatives.

Before moving on, some auxiliary variables are given.
Let the $i$th manipulator's estimated value of $a_0(t)$ be $a_i(t) \in {\mathbb R}^m$, $\forall i \in \mathcal V$.
For any $i,j \in \mathcal V$ and $t \in [t_s,t_{s+1})$, some auxiliary variables are defined as follows:
\begin{equation}\label{2.1}
  \left\{
  \begin{array}{lll}
    \bar q_{ij} = q_i - q_j - \eta_{si} + \eta_{sj},\\ \\
    \bar {\dot q}_{ij} = \dot q_i - \dot q_j,\\ \\
    \bar a_{ij} = a_i - a_j.
  \end{array}
  \right.
\end{equation}
\begin{remark}
The variable $\bar q_{ij}$ presented in (\ref{2.1}) contains the information of the time-varying formations and switches at the time sequence
$\{ t_1,t_2,\ldots,t_s,\ldots \}$. Besides, $\bar q_{ij} = 0$ means that the formation described by $\digamma _s$
is obtained for the MMS.
\end{remark}
Let $\bar q_{i} = q_i - \eta_{si} - x_0$, $\bar {\dot q}_{i} = {\dot q}_i - v_0$, $\bar a_{i} = a_i - a_0$ and
\begin{equation}\label{1.5}
\begin{array}{lll}
  \ddot q_{ri} = a_i - \varphi(sig(\sum\limits_{j \in {\mathcal N}_i}  {w_{ij}} \bar q_{ij} + p_i \bar q_{i})^{\alpha_1}) \\
  ~~~~~~~~~~ - \psi(sig(\sum\limits_{j \in {\mathcal N}_i}  {w_{ij}} \bar {\dot q}_{ij} + p_i \bar {\dot q}_{i})^{\alpha_2}),
  \end{array}
\end{equation}
where $\alpha_1, \alpha_2 > 0$ are positive constants, $\varphi$ and $\psi$ are continuous odd vector fields
satisfying $z^T\varphi(z) > 0$, $z^T\psi(z) > 0$ $(\forall z \neq 0)$,
$\varphi(z) = c_1 z + o(z)$ and $\psi(z) = c_2 z + o(z)$ around $z = 0$ for some positive constants $c_1$ and $c_2$,
$w_{ij}$ is the $(i,j)$th entry of the adjacency matrix $\mathcal W$, $p_i$ is the weight between the leader and the $i$th manipulator,
$sig(z)^\kappa = {\rm col}\{ |z_1|^{\kappa} sign(z_1), \cdots, |z_m|^{\kappa} sign(z_m) \}$,
$sign(\cdot)$ is the signum function, $\forall \kappa \in \mathbb R, z \in {\mathbb R}^m$.
We then propose the following distributed estimator-based control
\begin{subequations}\label{1.6}
\begin{equation}\label{1.7}
   ~~~~~\tau_i = {\mathcal H}_i(q_i)\ddot q_{ri} + {\mathcal C}_i(q_i,\dot q_i){\dot q_i} + g_i(q_i),
\end{equation}
\begin{equation}\label{1.8}
    \dot a_i = -\beta sgn(\sum\limits_{j \in {\mathcal N}_i}  {w_{ij}} \bar a_{ij} + p_i \bar a_i),
\end{equation}
\end{subequations}
where $\beta$ is presented in Assumption \emph{\textbf{A2}},
$sgn(z) = {\rm col}\{ sign(z_1), \cdots, sign(z_m) \}$, $\forall z \in {\mathbb R}^m$.

\begin{remark}
As shown in (\ref{1.5}), the sliding-mode estimator (\ref{1.8}) provides a distributed estimated value
$a_i$ to construct the auxiliary variable $q_{ri}$.
Moreover, inspired by the inverse dynamics control technology proposed in \cite{Spong}-\cite{SuZ},
the input torque $\tau_i$ presented in (\ref{1.7}) is developed by using $q_{ri}$.
Thus, the control law (\ref{1.6}) is called distributed estimator-based control.
\end{remark}

\begin{theoremx}\label{t1}
Suppose that Assumptions \textbf{A1} and \textbf{A2} hold. Using (\ref{1.6}) for (\ref{1.1}),
if $0 < \alpha_1  < 1$ and $\alpha_2  = 2\alpha_1 /(\alpha_1 + 1)$, then (\ref{1.2}) holds, $i.e.$, the time-varying formation tracking is achieved for the {MMS}.
\end{theoremx}

\begin{proof}
The proof proceeds in the following three steps.
First, the simplification of the close-loop system is derived from the {finite-time stability of sliding-mode estimators}.
Second, the global asymptotic stability is proved based on the {LaSalle's Invariance Principle}.
Thirdly, the global finite-time stability is demonstrated using {finite-time stability arguments for homogeneous systems}.

For the first presentation, the simplification of the close-loop system is carried out.
Substituting (\ref{1.5}) and (\ref{1.7}) into (\ref{1.1}) gives
\begin{equation}\label{1.9}
\begin{array}{lll}
  {\mathcal H}_i(q_i)[\ddot q_i - a_i + \varphi(sig(\sum\limits_{j \in {\mathcal N}_i}  {w_{ij}} \bar q_{ij} + p_i \bar q_{i})^{\alpha_1}) \\
  ~~~~~~~~~~~~~~~~~ + \psi(sig(\sum\limits_{j \in {\mathcal N}_i}  {w_{ij}} \bar {\dot q}_{ij} + p_i \bar {\dot q}_{i})^{\alpha_2})] = 0.\\
  \end{array}
\end{equation}
The positive definiteness of ${\mathcal H}_i(q_i)$ implies that the eigenvalues of ${\mathcal H}_i(q_i)$ is greater than $0$.
Then the combination of (\ref{1.8}) and (\ref{1.9}) yields the following \textsl{cascade} system:
\begin{equation}\label{2.0}
 \begin{array}{lll}
  \ddot q_i = a_i - \varphi(sig(\sum\limits_{j \in {\mathcal N}_i}  {w_{ij}} \bar q_{ij} + p_i \bar q_{i})^{\alpha_1}) \\
  ~~~~~~~~~ - \psi(sig(\sum\limits_{j \in {\mathcal N}_i}  {w_{ij}} \bar {\dot q}_{ij} + p_i \bar {\dot q}_{i})^{\alpha_2}),\\ \\
  \dot a_i = -\beta {sgn}(\sum\limits_{j \in {\mathcal N}_i}  {w_{ij}} \bar a_{ij} + p_i \bar a_i),
  \end{array}
\end{equation}
Let $\bar a$ be the column stack vector of $\bar a_i$, $\forall i \in \mathcal V$.
The sliding-mode estimator (\ref{1.8}) can be rewritten as
\begin{equation}\label{0.9}
  \dot {\bar a} = -\beta {sgn}(\mathcal{M} \bar a) - 1_n \otimes a_0,
\end{equation}
where $1_n$ denotes the \textit{n}-dimensional column vector whose elements are all one.
By {Lemma \ref{l5}}, $\mathcal{M}$ is symmetric positive definite.
Take the Lyapunov function ca.
.
ndidate $V_0 = 1/2 {\bar a}^T \mathcal{M} \bar a$ for system (\ref{0.9}).
By the similar analysis in {Theorem 3.1} of {\cite{fmg09}}, we get that
\begin{equation*}
    \dot V_0 \leq - (\beta - \mathop {\sup }\nolimits_{t \in \mathcal J} \| \dot a_0(t) \|)
    \frac {\lambda_{\min}(\mathcal M)\sqrt {2V_0}}{\sqrt {\lambda_{\max}(\mathcal M)}}.
\end{equation*}
Therefore, for the sliding-mode estimator (\ref{1.8}), there exists a bounded settle time given by
\begin{equation*}
  T_f = t_0 + \frac {\sqrt {2\lambda_{\max}(\mathcal M) V_0(t_0)}} {\lambda_{\min}(\mathcal M)(\beta - \mathop {\sup }\nolimits_{t \in \mathcal J} \| \dot a_0(t) \|)}
\end{equation*}
such that $a_i = a_0$ when $t \geq T_f$, $\forall i \in \mathcal{V}$.
We then show that for bounded initial values $q_i(t_0)$ and $\dot q_i(t_0)$,
invoking (\ref{1.6}) for (\ref{1.1}),
the states $q_i(t)$ and $\dot q_i(t)$ remain bounded when $t \in [t_0,T_f]$, $\forall i \in \mathcal V$.
The distributed sliding-mode estimator (\ref{1.8}) implies that $a_i(t)$ remain bounded
for any initial value $a_i(t_0)$ when $t \in [t_0,T_f]$.
For bounded states $q_i$ and $\dot q_i$, $\forall i \in \mathcal V$, equation (\ref{2.1}) implies that
$\bar q_i$, $\bar {\dot q}_i$, $\bar q_{ij}$ and $\bar {\dot q}_{ij}$ remain bounded when $t \in [t_0,T_f]$, $\forall j \in \mathcal V$.
It thus follows from (\ref{2.0}) that $\ddot q_i$ is bounded with respect to bounded states
$a_i$, $\bar q_i$, $\bar {\dot q}_i$, $\bar q_{ij}$ and $\bar {\dot q}_{ij}$.
Thus, we can obtain that $q_i(t)$ and $\dot q_i(t)$ remain bounded for bounded initial values
$q_i(t_0)$ and $\dot q_i(t_0)$ when $t \in [t_0,T_f]$, $\forall i \in \mathcal V$.
Thus, using (\ref{1.5}) and (\ref{1.6}) for (\ref{1.1}), when $t \geq T_f$, the closed-loop dynamics of system (\ref{1.1}) can be rewritten as
\begin{equation}\label{2.3}
 \begin{array}{lll}
  \bar {\ddot q}_i = - \varphi(sig(\sum\limits_{j \in {\mathcal N}_i}  {w_{ij}} \bar q_{ij} + p_i \bar q_{i})^{\alpha_1}) \\
  ~~~~~ - \psi(sig(\sum\limits_{j \in {\mathcal N}_i}  {w_{ij}} \bar {\dot q}_{ij} + p_i \bar {\dot q}_{i})^{\alpha_2}),
  \end{array}
\end{equation}
where $\bar {\ddot q}_i = {\ddot q}_i - a_0$. It thus follows from (\ref{2.1}) that $\bar {\dot q}_i$ and $\bar {\ddot q}_i$ are the first-order
and second-order derivatives of $\bar q_i$, $\forall i \in \mathcal V$.
Let $\bar q$, $\bar {\dot q}$ and $\bar {\ddot q}$ be the column stack vectors of $\bar q_i$, $\bar {\dot q}_i$ and $\bar {\ddot q}_i$ respectively,
$\forall i \in \mathcal V$. System (\ref{2.3}) can be rewritten as
\begin{equation}\label{2.2}
 \begin{array}{lll}
  \bar {\ddot q} = - \varphi(sig(\mathcal{M} \bar q)^{\alpha_1} )- \psi(sig(\mathcal{M} \bar {\dot q})^{\alpha_2}).
  \end{array}
\end{equation}
The first presentation shows that for bounded initial values $a_i(t_0)$, $q_i(t_0)$ and $\dot q_i(t_0)$,
the states $a_i(t)$, $q_i(t)$ and $\dot q_i(t)$ remain bounded when $t \in [t_0,T_f]$,
and the close-loop dynamics of (\ref{1.1}) under the control algorithms (\ref{1.5}) and (\ref{1.6})
is equivalent to equation (\ref{2.2}) when $t \geq T_f$.
\par
For the second presentation, the global asymptotic stability of system (\ref{2.2}) is analyzed.
Let an auxiliary variable $y = \mathcal{M} \bar q \in {\mathbb R}^{mn}$. Then $\dot y = \mathcal{M} \bar {\dot q}$ and $\ddot y = \mathcal{M} \bar {\ddot q}$.
When $t \geq T_f$, for (\ref{2.2}), consider the {Lyapunov function candidate} $V = V_1 + V_2$ with
\begin{equation*}
 \begin{array}{lll}
  V_1 = \sum\limits_{k = 1}^{mn} {\int_0^{y(k)} {\varphi (sig{(\sigma)^{\alpha_1 }})} {\kern 1pt} d\sigma},\\ \\
  V_2 = \frac{1}{2} \bar {\dot q}^T \mathcal{M} \bar {\dot q},
  \end{array}
\end{equation*}
where ${y(k)} \in \mathbb R\ (k = 1,2, \ldots ,mn)$ denotes the $k$th element of the vector $y$.
By {Lemma \ref{l5}}, $\mathcal{M} $ is symmetric positive definite.
It thus follows from the definition of $\varphi (\cdot)$ that the {Lyapunov function candidate} $V$ is positive definite.
Taking the derivatives of $V_1$ and $V_2$ along (\ref{2.2}) renders that
\begin{equation*}
 \begin{array}{lll}
  \dot V_1 &=& \sum\limits_{k = 1}^{mn} \dot y(k) \varphi (sig{[y(k)]^{\alpha_1 }}) \\
  &=& \dot y^T \varphi (sig(y)^{\alpha_1 }), \\ \\
  \dot V_2 &=& \bar {\dot q}^T \mathcal{M} \bar {\ddot q} \\
  &=& - \dot y^T \varphi(sig(y)^{\alpha_1} )- \dot y^T \psi(sig(\dot y)^{\alpha_2}),
  \end{array}
\end{equation*}
It thus follows that
\begin{equation*}
 \begin{array}{lll}
  \dot V &=& \dot V_1 + \dot V_2 \\
  &=& - \dot y^T \psi(sig(\dot y)^{\alpha_2}).
  \end{array}
\end{equation*}
Considering that $\psi(\cdot)$ is continuous odd function, we can conclude that $\dot V \leq 0$.
Besides, $\dot V = 0$ gives that $\dot y = 0$.
It thus follows from the positive definiteness of $\mathcal M$ that $\dot V = 0$ if and only if $\bar {\dot q} = 0$,
which implies that $\bar {\ddot q} = 0$. It thus follows from (\ref{2.2}) that $\varphi(sig(\mathcal{M} \bar q)^{\alpha_1} ) = 0$,
which means that $\bar q = 0$. By {LaSalle's Invariance Principle} in {Lemma \ref{l1}},
for any bounded $\bar q(T_f)$ and $\bar {\dot q}(T_f)$, the states $\bar q \to 0$ and $\bar {\dot q} \to 0$ as $t \to \infty$.
Hence, the second presentation shows that the equilibrium $(\bar q = 0,\bar {\dot q} = 0)$ of system (\ref{2.2}) is global asymptotic stable.
\par
For the third presentation, the global finite-time stability of system (\ref{2.2}) is analyzed.
First, the local finite-time stability is proven by invoking {Lemma \ref{l2}} and {\ref{l3}}.
To this end, let $z_1 = \bar q$, $z_2 = \bar {\dot q}$ and $z = {\rm col}(z_1,z_2)$.
By the definition of $\varphi(\cdot)$ and $\psi(\cdot)$ right after (\ref{1.5}), we can get that system (\ref{2.2})
can be written as
\begin{equation}\label{2.4}
\left\{ \begin{array}{lll}
  {\dot z_1} = z_2,\\
  {\dot z_2} = f(z_1,z_2) + {\tilde f}(z_1,z_2),
\end{array} \right.
\end{equation}
where
\begin{equation*}
\left\{ \begin{array}{lll}
  f(z_1,z_2) = - c_1 sig(\mathcal{M} z_1)^{\alpha_1} - c_2 sig(\mathcal{M} z_2)^{\alpha_2},\\
  {\tilde f}(z_1,z_2) = - o( sig(\mathcal{M} z_1)^{\alpha_1} ) - o( sig(\mathcal{M} z_2)^{\alpha_2} ).
  \end{array}\right.
\end{equation*}
It visibly follows that ($z_1 = 0,z_2 = 0$) is the equilibrium of system (\ref{2.4}).
Considering that $\alpha_2 = 2\alpha_1/(\alpha_1 + 1)$, we can conclude that system (\ref{2.4}) is {locally homogeneous}
of degree $\lambda  = \alpha_1  - 1<0$ with respect to dilation ${\rm col}(2_{mn},(\alpha_1 + 1)_{mn})$,
where $2_{mn}$ and $(\alpha_1 + 1)_{mn}$ are {$mn$-dimensional} column vectors whose elements are $2$ and $\alpha_1 + 1$ respectively.
Hence, the third presentation shows that the equilibrium $(\bar q = 0,\bar {\dot q} = 0)$ of system (\ref{2.2}) is finite-time asymptotic stable.
\par
By {Lemma \ref{l3}}, the second and third presentations show that for bounded $\bar q(T_f)$ and $\bar {\dot q}(T_f)$,
there exists a time point $\bar{T}_f > T_f$ that the states $\bar q \to 0$ and $\bar {\dot q} \to 0$ as $t \to \bar{T}_f$.
By the first presentation, $\bar q(T_f)$ and $\bar {\dot q}(T_f)$ remain bounded for bounded initial value $\bar q(t_0)$,
$\bar {\dot q}(t_0)$ and $a_i(t_0)$. Hence, for bounded initial value $\bar q(t_0)$, $\bar {\dot q}(t_0)$
and $a_i(t_0)$, the states $\bar q \to 0$ and $\bar {\dot q} \to 0$ as $t \to \bar{T}_f$. This completes the proof.
\end{proof}

Note that the following \textsl{necessary and sufficient condition} can be easily obtained by some simple transformation for {Theorem \ref{t1}}.

\begin{corollaryx}\label{c1}
Suppose that $0 < \alpha_1  < 1$, $\alpha_2  = 2\alpha_1 /(\alpha_1 + 1)$, and Assumption \textbf{A2} holds. Using (\ref{1.5}) and (\ref{1.6}) for (\ref{1.1}),
then (\ref{1.2}) holds ($i.e.$, the time-varying formation tracking is achieved for the {MMS}) if and only if Assumption \textbf{A1} holds.
\end{corollaryx}

\begin{proof}
The sufficiency of {Corollary \ref{c1}} is proved as the same as in {Theorem \ref{t1}}.
Next we show the necessity part by contradiction.
If {Assumption \emph{\textbf{A1}}} dose not hold, there exists an isolated subset of manipulators,
which cannot obtain any information of the leader directly or mediately.
It follows that the evolution of the close-loop dynamics of these manipulators is carried out without any information of the leader.
Thus, these manipulators cannot necessarily follow the trajectory of the leader.
This ends the proof.
\end{proof}

Let a switching graph ${\mathcal G}(t) = \{ \mathcal V, \mathcal E(t), \mathcal W(t) \}$ describe the interaction of the {MMS}),
where $\mathcal W(t) = [w_{ij}(t)]_{n \times n}$ represents the weight adjacency matrix.
Let ${\mathcal P}(t) = [{p_1}(t),{p_2}(t),\ldots,{p_n}(t)]^T$ be the switching nonnegative weight vector
between the \textit{n} nodes and the leader.
Then the following corollary can be obtained for the case, in which the communication topology is switching.

\begin{corollaryx}\label{c2}
Suppose that \textbf{A2} holds and the leader are reachable to the MMS under ${\mathcal G}(t)$ and ${\mathcal P}(t)$. Let the control algorithms be replaced by
\begin{equation*}
\left\{ \begin{array}{lll}
  \ddot q_{ri} = a_i - \varphi(sig(\sum\limits_{j \in {\mathcal N}_i}  {w_{ij}}(t) \bar q_{ij} + p_i(t) \bar q_{i})^{\alpha_1}) \\
  ~~~~~~~~~ - \psi(sig(\sum\limits_{j \in {\mathcal N}_i}  {w_{ij}}(t) \bar {\dot q}_{ij} + p_i(t) \bar {\dot q}_{i})^{\alpha_2}),\\ \\
  \tau_i = {\mathcal H}_i(q_i)\ddot q_{ri} + {\mathcal C}_i(q_i,\dot q_i){\dot q_i} + g_i(q_i),\\ \\
  \dot a_i = -\beta sgn(\sum\limits_{j \in {\mathcal N}_i}  {w_{ij}}(t) \bar a_{ij} + p_i(t) \bar a_i),
  \end{array} \right.
\end{equation*}
then (\ref{1.2}) holds ($i.e.$, the time-varying formation tracking is achieved for the {MMS}) if and only if Assumption \textbf{A1} holds.
\end{corollaryx}

\begin{proof}
The proof can be easily derived by the combination of {Theorem \ref{t1}} and {Lemma 6} presented in \cite{WangHong}, and is omitted here.
\end{proof}

\begin{remark}
Note that the functions ${\varphi }( \cdot)$ and ${\psi }( \cdot)$ can be easily selected, such as
$x$, $sat(x)$, and $tanh(x)$, where $sat(\cdot)$ and $tanh(\cdot)$ denote the saturation function and the hyperbolic tangent function respectively.
Besides, by the boundedness of $sat(\cdot)$ and $tanh(\cdot)$, we can conclude that the control law in this paper is bounded by the boundedness
of the dynamic terms in system (\ref{1.1}).
\end{remark}

\begin{remark}
The dynamics of the leader can also be described by the Euler-Lagrange equation
${\mathcal H}_0(q_0){\ddot q}_0 + {\mathcal C}_0(q_0,\dot q_0){\dot q_0} + g_0(q_0) = \tau_0$,
which gives a additional task for designing $\tau_0$.
In this case, the MMS has a master-slave structure, in which the master manipulator acts as the leader
while the slave manipulators act as followers {\cite{CWRHT,WCHTY}}.
By designing suitable $\tau_0$ such that Assumption \textbf{A2} holds following \cite{Mei01}, the main results presented
in this paper can still be effective.
\end{remark}

\begin{remark}
Comparing with \cite{Mei01,ChenFeng}, in which {global asymptotic stability} is achieved, we study the {global finite-time stability} for time-varying formation tracking which is more practical and challenging than traditional {global asymptotic stability}, especially for robotic systems.
Different from \cite{Wang01,Wang02}, in which the constant agreement value is taken into account,
we consider the time-varying formation tracking problem of multi-robot systems with a dynamic leader.
\end{remark}

\section{\sc Simulations}
In this section, simulations are presented to illustrate the effectiveness of the proposed algorithms.
We consider the time-varying formation tracking problem for a {MMS} containing six manipulators ($i.e.$, agents) with three desired formations.
Each agent is assumed to be a planer robotic manipulator with two revolute joints, $i.e.$, $q_i \in \mathbb{R}^2,~\forall i \in \mathcal V$.
The dynamic model and the physical parameters presented in {\cite{SuZ}} are invoked.
For simplify, in our simulation, we choose $w_{ij}=1$ if agent $i$ can access the information of agent $j$,
$w_{ij}=0$ otherwise; $p_{i}=1$ if agent $i$ can obtain the information of the leader directly, $p_i=0$ otherwise.
The interaction topology is shown in Fig.\ref{Fig.0}. The {Laplacian} matrix $\mathcal L$ is
\begin{equation*}
  {\mathcal L} =
  \left[ {\begin{array}{*{20}{c}}
{\begin{array}{*{20}{c}}
\begin{array}{l}
2\\
 0\\
-1\\
-1\\
0\\
0
\end{array}&\begin{array}{l}
0\\
0\\
0\\
0\\
0\\
0
\end{array}&\begin{array}{l}
-1\\
 0\\
1\\
 0\\
0\\
0
\end{array}&\begin{array}{l}
-1\\
 0\\
0\\
 1\\
0\\
0
\end{array}
\end{array}}&\begin{array}{l}
0\\
0\\
0\\
0\\
1\\
-1
\end{array}&\begin{array}{l}
0\\
0\\
0\\
0\\
-1\\
1
\end{array}
\end{array}} \right],
\end{equation*}
and the nonnegative weight vector is given by $\mathcal P = {\rm col}(1,1,0,0,1,0)$.
The elements of the initial values $q_i(0)$, $\dot q_i(0)$ and $a_i(0)$ are randomly selected from $[-6,6]$.

The finite set of desired formations ${\digamma_c} = \{ {\digamma _0},{\digamma _1}, {\digamma _2} \}$ is shown in Fig.\ref{Fig.1},
the local coordinates in the 2D plane is given as ${\digamma_s}  = \{ {\eta _{s1}},{\eta _{s2}}, \ldots ,{\eta _{s6}} \}$,
$s =0, 1, 2$. The details of $\eta_{si} = [a_{si}, b_{si}]^T$, $i= 1, \ldots, 6$, are presented in Table \ref{tab.1}.
The sampling period is adopted to be $10$ $ms$. The simulation time span is selected as $t \in [0,50]$.
The time-varying formation $\digamma(t)$ are given as following
\begin{equation*}
  \digamma(t)=
  \left\{
    \begin{array}{lll}
    \digamma_0, ~~t \in [0,15),\\
    \digamma_1, ~~t \in [15,35),\\
    \digamma_2, ~~t \in [35,50).
    \end{array}
  \right.
\end{equation*}
The trajectory of the leader is given by
$x_0(t) = {\rm col}[ 30\cos (0.05\pi t),30\sin (0.05\pi t) ]$, and $v_0(t)$, $a_0(t)$ can be calculated easily,
where $t \geq 0$.
Without loss of generality, let $\varphi(z) = 100z$ and $\psi(z) = 100z$.
The other control parameters are selected as follows: $\alpha_1 = 0.2$, $\beta = 4$, and $\alpha_2$ can be easily computed.

\begin{table}[h]
 \centering
 \caption{The local coordinates. }\label{tab.1}
 \begin{tabular}{cccc}
  \hline
   $a_{si},b_{si}$ & $s = 0$ & $s = 1$ & $s = 2$ \\
  \hline
   $i = 1$ & $1, \sqrt 3$ & $2, \sqrt 3$ & $2/3, \sqrt 3 /2$ \\
   $i = 2$ & $2, 0$ & $2, 0$ & $8/3, 0$ \\
   $i = 3$ & $1, - \sqrt 3$  & $2, - \sqrt 3$ & $2/3, - \sqrt 3 /2$ \\
   $i = 4$ & $- 1, - \sqrt 3$ & $- 2, - \sqrt 3$ & $- 4/3, - \sqrt 3$ \\
   $i = 5$ & $- 2, 0$ & $- 2,0$ & $- 4/3, 0$ \\
   $i = 6$ & $- 1, \sqrt 3$  & $- 2, \sqrt 3$ & $- 4/3, \sqrt 3$ \\
  \hline
 \end{tabular}
\end{table}

\begin{figure}[H]
  \centering
  % Requires \usepackage{graphicx}
  \includegraphics[width=5cm]{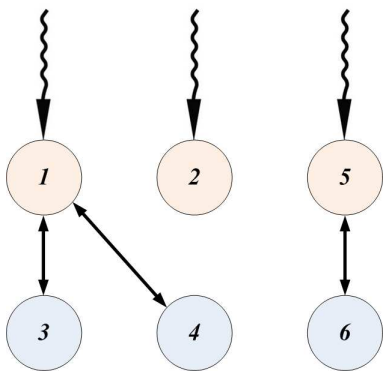}\\
  \caption{The interaction graph $\mathcal G$, where the agent $1$, $2$, and $5$ can access the information of the leader directly. }\label{Fig.0}
\end{figure}

\begin{figure}[H]
  \centering
  % Requires \usepackage{graphicx}
  \includegraphics[width=11cm]{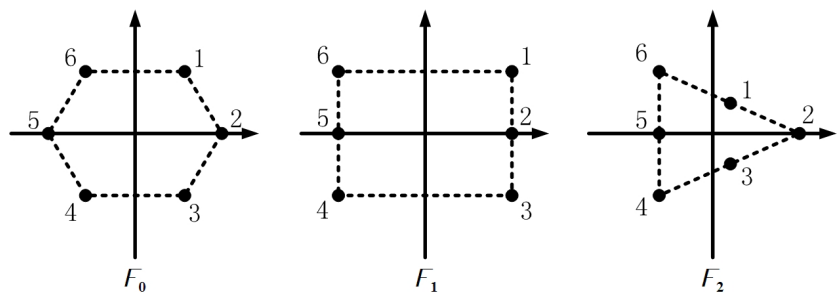}\\
  \caption{The formations from left to right are $\digamma_0$, $\digamma_1$ and
  $\digamma_2$ respectively. The black point $i$ denote the robot $i$ in the local coordinate. }\label{Fig.1}
\end{figure}

\begin{figure}[H]
  \centering
  % Requires \usepackage{graphicx}
  \includegraphics[width=11cm]{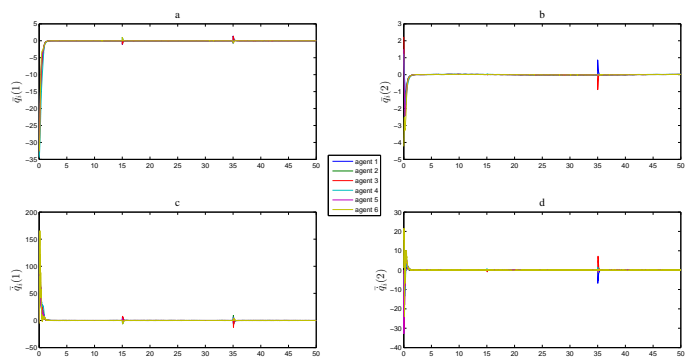}\\
  \caption{Trajectories of ${\bar q}_i$ and ${\bar {\dot q}}_i$ under $\mathcal G$.}\label{Fig.3}
\end{figure}

\begin{figure}[H]
  \centering
  \includegraphics[width=9cm]{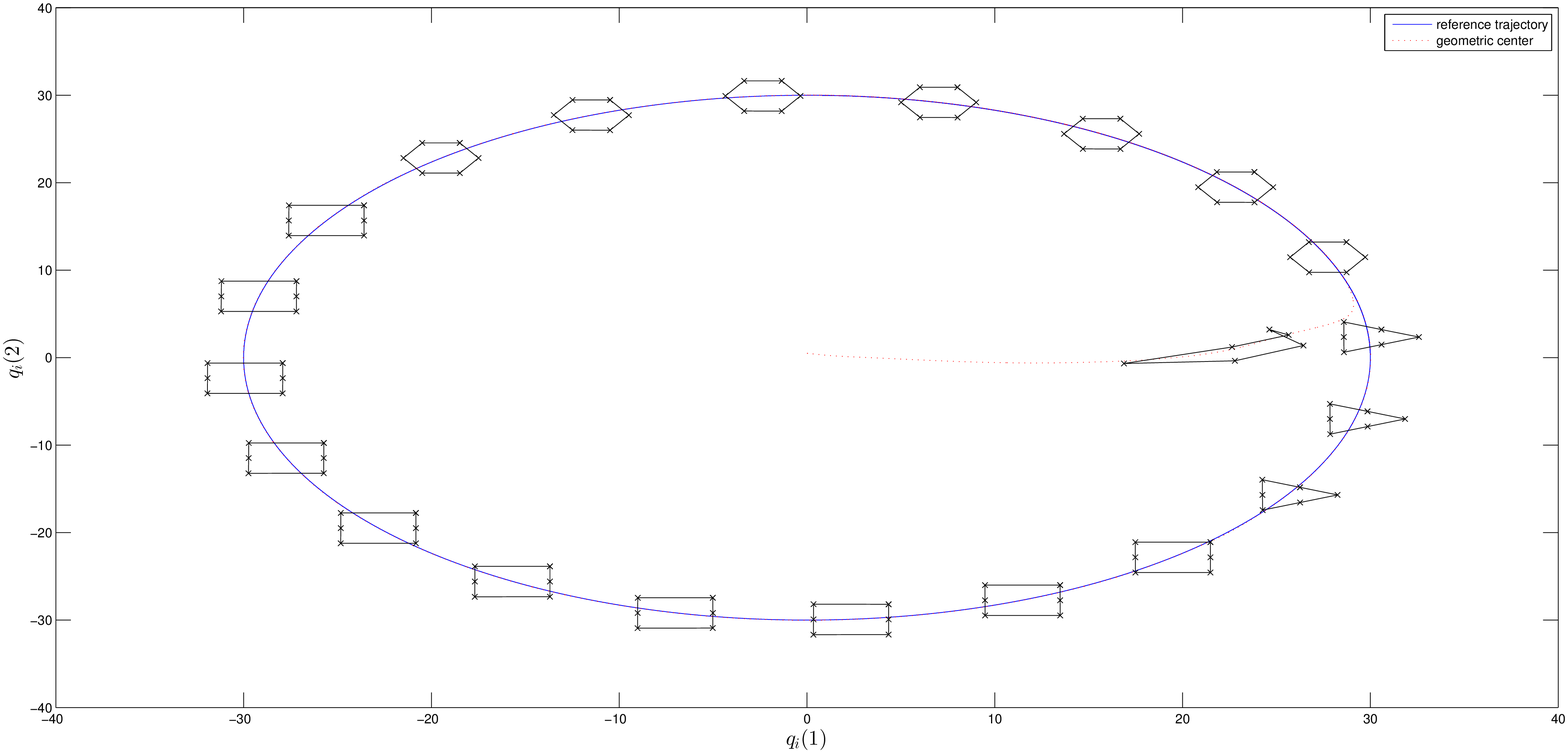}\\
  \caption{Trajectories and the formation of the six robots under $\mathcal G$.}\label{Fig.5}
\end{figure}

The simulation results are presented in Fig.\ref{Fig.3} and Fig.\ref{Fig.5}.
Fig.\ref{Fig.3} shows that the tracking errors ${\bar q}_i$ and
${\bar {\dot q}}_i$ defined in (\ref{2.1}) converge to zero in finite time at each dwell time interval,
which means the time-varying formations of the {MMS} in the 2D plane and the tracking of the leader can
be achieved simultaneously, $i.e.$, the {time-varying formation tracking} is accomplished.
Additionally, the trajectory of the manipulators in 2D space is illustrated in Fig.\ref{Fig.5}.
It follows that the robots can reach the desired time-varying formation and
the geometric center of the {MMS} follows the leader as required.
It is clear in Fig.\ref{Fig.3} and Fig.\ref{Fig.5} that using the control algorithm (\ref{1.6})
under the aforementioned configurations, the time-varying formation tracking can be achieved for the MMS.

\begin{remark}
It is shown from picture $b$ and $d$ in Fig.\ref{Fig.3} that the second elements of ${\bar q}_i$ and
${\bar {\dot q}}_i$ do not change at the switching time instant $t = 15$.
Note that the time-varying formation $\digamma(t)$ changes from $\digamma_0$ to $\digamma_1$.
By the set of $\digamma_0$ and $\digamma_1$ in Table 1, $b_{si}$ stays the same at the switching time instant
$t = 15$, which thus gives that ${\bar q}_i$ and ${\bar {\dot q}}_i$ do not change at the switching time instant.
\end{remark}

\section{Conclusion}
For multiple manipulator systems ({MMSs}) under fixed and switching graphs, the time-varying formation tracking
problem is addressed using inverse dynamics control technologies.
Based on the functional characteristics of {MMSs}, an explicit formulation of time-varying formation is presented.
The conditions (including sufficient conditions, necessary and sufficient conditions) on the interaction topology and
control parameters are derived.
Simulation results are presented to verify the effectiveness of the proposed algorithms.
A few interesting issues, which are not addressed in this paper, concern the time-varying formation tracking problems
of uncertain Euler-Lagrange systems and the extension of the presented approaches to the case of the polynomial trajectories.
These issues will be considered in our future work.

%\begin{window}[{\includegraphics[width=1in,height=1.25in,clip,keepaspectratio]{gemingfeng.eps}}]{Ming-Feng Ge}
%was born in Hubei, China, in
%1986. He graduated in automation from Huazhong
%University of Science and Technology, Wuhan,
%China, in 2008. Currently, he is working towards
%the Ph.D. degree at the College of Automation,
%Huazhong University of Science and Technology, Wuhan, China.
%\par
%His research interests include complex dynamical networks, impulsive
%and hybrid control systems, cooperative control of networked robotic systems.
%\end{window}

\end{document}